\documentclass[conference,letterpaper]{IEEEtran}

\usepackage{cite}

\ifCLASSINFOpdf
            \else
                  \fi

\usepackage{amsmath}
\usepackage{amsfonts}

\usepackage{graphicx}
\graphicspath{{./figures/}} 
\usepackage{datetime}

\usepackage{microtype}\usepackage{comment}
\usepackage{enumitem} \usepackage{todonotes}
\usepackage{soul}
\usepackage[framemethod=TikZ]{mdframed}
\usepackage{listings}
 \usepackage{setspace}

\usepackage{mathtools}
\DeclarePairedDelimiter{\ceil}{\lceil}{\rceil}

\newenvironment{framed}{\begin{mdframed}[backgroundcolor=blue!5,outerlinewidth=0pt,outerlinecolor=blue!5,innerleftmargin=10pt,innerrightmargin=0pt,innertopmargin=-10pt,innerbottommargin=-0pt]}{\end{mdframed}}

\newtheorem{lemma}{Lemma}
\newtheorem{theorem}{Theorem}
\newtheorem{definition}{Definition}
\newtheorem{proof}{Proof}

\newcount\outlineon
\newcommand{\outline}[1]{\ifnum\outlineon>0\
	\\\noindent\fbox{\begin{minipage}{\linewidth}{\footnotesize
				#1}\end{minipage}}\\\\\fi}
\newcommand{\exclude}[1]{}

\renewcommand{\baselinestretch}{0.92}

\begin{document}
\title{Analysis of Dynamic Memory Bandwidth Regulation in Multi-core Real-Time Systems}

\author{
  \IEEEauthorblockN{
    Ankit Agrawal\IEEEauthorrefmark{1},
    Renato Mancuso\IEEEauthorrefmark{2},
    Rodolfo Pellizzoni\IEEEauthorrefmark{3}, 
    Gerhard Fohler\IEEEauthorrefmark{1}
  }
  \IEEEauthorblockA{
    \IEEEauthorrefmark{1}Technische Universit{\"a}t Kaiserslautern, Germany, 
    \{agrawal, fohler\}@eit.uni-kl.de
  }
  \IEEEauthorblockA{
    \IEEEauthorrefmark{2}Boston University, USA,
    rmancuso@bu.edu
  }
  \IEEEauthorblockA{
    \IEEEauthorrefmark{3}University of Waterloo, Canada,
    rpellizz@uwaterloo.ca
  }
}

\outlineon=0

\maketitle

\begin{abstract}

One of the primary sources of unpredictability in modern multi-core
embedded systems is contention over shared memory resources, such as
caches, interconnects, and DRAM. Despite significant achievements in
the design and analysis of multi-core systems, there is a need for a
theoretical framework that can be used to reason on the worst-case
behavior of real-time workload when both processors and memory
resources are subject to scheduling decisions. 

In this paper, we focus our attention on dynamic allocation of main memory bandwidth.
In particular, we study how to determine the worst-case response time of tasks
spanning through a sequence of time intervals, each with a different bandwidth-to-core assignment.
We show that the response time computation can be reduced to a maximization problem 
over assignment of memory requests to different time intervals, and we provide an 
efficient way to solve such problem. As a case study, we then demonstrate
how our proposed analysis can be used to improve the schedulability
of Integrated Modular Avionics systems in the presence of memory-intensive 
workload.

\end{abstract}
 

\IEEEpeerreviewmaketitle

\makeatletter
\let\@ORGmakecaption\@makecaption
\long\def\@makecaption#1#2{\@ORGmakecaption{#1}{#2}\vskip\belowcaptionskip\relax}
\makeatother

\section{Introduction}

Over the last decade, multi-core systems have rapidly increased in
popularity and they are now the de-facto standard in the embedded
computing industry. Multi-core systems are significantly more
challenging to analyze compared to their single-core counterparts due
to the extensive sharing of hardware resources among logically
independent execution flows. The primary source of performance
unpredictability, in this class of systems, can be identified as the
memory hierarchy. In fact, the memory hierarchy in multi-core
platforms is comprised of a number of components that are concurrently
accessed by multiple cores. These include: multi-level CPU caches,
shared memory controllers and DRAM banks, and shared I/O devices. The
interplay of accesses originated by multiple cores has a direct impact
on the timing of subsequent memory accesses. The resulting temporal
variability is in the range of multiple orders of magnitude, meaning that
inaccurate performance modeling and analysis can lead to overly
pessimistic worst-case execution time (WCET) estimates.

Despite the remarkable achievements in the analysis of hard real-time
workload on multi-core systems, there is a fundamental lack of
self-contained theoretical frameworks that can be used to reason on
the schedulability of a generic multi-core hard real-time workload
when both CPU and memory resources are subject to scheduling
decisions. In fact, while consolidated techniques are used to reason
about CPU scheduling, comparatively less general results are available
to reason on memory scheduling. An even slimmer body of works has
provided general results to reason on co-scheduling of CPU and
memory. The majority of works in this area assume fixed assignment of
memory resources to CPUs.

{\bf Memory Scheduling:} there are two dimensions to the problem of
assigning memory resources to applications. The first dimension is
\emph{space} scheduling, concerning the allocation over time of memory
space (e.g., cache lines, DRAM banks, scratchpad pages). A second
dimension is \emph{temporal} scheduling, i.e., scheduling of access to
a shared memory interface (e.g., an interconnect, a bus, or a memory
controller). \emph{In this paper, we focus on the temporal dimension
  of memory scheduling.} In a nutshell, memory interfaces/subsystems
are associated with a characteristic sustainable bandwidth that can be
partitioned among the CPUs of a multi-core system. If the
bandwidth-to-cores assignment is determined offline and remains unchanged
over time, we say that memory bandwidth is {\bf statically
  partitioned}. Conversely, if bandwidth is {\bf dynamically assigned}
to cores, we say that memory bandwidth is subject to
\emph{scheduling}. We hereafter interchangeably use the terms ``memory
scheduling'', ``memory bandwidth scheduling'', or ``bandwidth
partitioning'' referring to the same concept.

Since memory bandwidth is constrained and often represents a
bottleneck in multi-core systems, memory scheduling is an important
dimension to consider and a way to achieve important real-time
performance improvements. Clearly, if the memory bandwidth assigned to
each core changes over time, this will have an effect on the response
time of tasks. In this case, how can the worst-case response time be
calculated? In this paper, we address this question. More specifically,
we study the problem of determining the worst-case response time for a
task that spans a sequence of time intervals, each with a
different bandwidth-to-cores assignment.

In this paper, we make the following contributions:

\begin{enumerate}
  \item we improve response time calculation under static (over time)
    but arbitrary (across cores) bandwidth partitioning;
  \item we provide a general framework to perform response time
    analysis under dynamic bandwidth partitioning. Our approach can be
    used to analyze memory schedulers as long as: (i) changes in
    bandwidth-to-cores allocation are time-triggered; or (ii) a
    critical instant can be found for the possible CPU-to-tasks and
    bandwidth-to-cores scheduling decisions;
  \item we demonstrate how the proposed analysis technique can be used
    in a time-triggered memory scheduling scenario, for an Integrated
    Modular Avionics (IMA) system. In particular, we show that 
    dynamic bandwidth allocation significantly
    outperforms static allocation in the presence of
    varying memory-intensive workload.
    \end{enumerate}

\newcommand{\Iconc}{\bar{I}}
\newcommand{\emusum}{\beta}
\newcommand{\jb}{\bar{j}}
\newcommand{\jt}{\tilde{j}}

\section{Background}
\label{sec:background_gen}

A memory interface is characterized by a \emph{maximum guaranteed
  bandwidth}. It is generally easy to analyze the temporal behavior of
memory requests when the interface operates below the maximum
guaranteed bandwidth~\cite{memguard_tc}. Conversely, if the rate of
memory requests exceeds such a threshold, the behavior of the memory
subsystem can be hard to analyze, or lead to overly pessimistic
worst-case estimates~\cite{cmu_dram}. In multi-core systems, however,
the available memory bandwidth can be arbitrarily distributed among
cores. Take a 2-core system for instance, as depicted in
Figure~\ref{fig:mem_part}. Workload on the two cores can be either
CPU-intensive (blue), or memory-intensive (red). For simplicity, the
figure assumes that CPU-intensive workload is unaffected by changes in
memory bandwidth (BW) assignment. Conversely, memory-intensive
workload is roughly linearly affected by it. An {\bf even} assignment
as depicted in Figure~\ref{fig:mem_part}(a) would provide 50\% of the
available memory bandwidth to each core. Even partitioning is not
flexible: mostly memory intensive workload is deployed on core Core~A,
while mostly CPU-intensive workload is scheduled on Core~B. As such,
workload is penalized on Core~A while memory bandwidth is wasted on
Core~B. Under this setup, the memory-intensive workload on Core~A and
B take 5 and 3 time units to complete, respectively. The overall
utilization is 95\%.

If Core~A is known to run memory-intensive tasks while Core~B mostly
processes CPU-intensive workload, it is beneficial to perform an {\bf
  uneven} assignment -- e.g., 80\% and 20\% of the available bandwidth
assigned to Core~A and B, respectively. This is depicted in
Figure~\ref{fig:mem_part}(b). In this case, bandwidth can be distributed
to better meet the CPU/memory needs of workload on the various
cores. The memory-intensive workload on Core~A can benefit from this
assignment, now completing in 2 time units. However, the (shorter)
memory-intensive workload on Core~B is negatively affected, completing
in 4 time units. Overall utilization decreases to 85\% in our example.

\begin{figure}
  \centering 
  \includegraphics[width=0.80\linewidth]{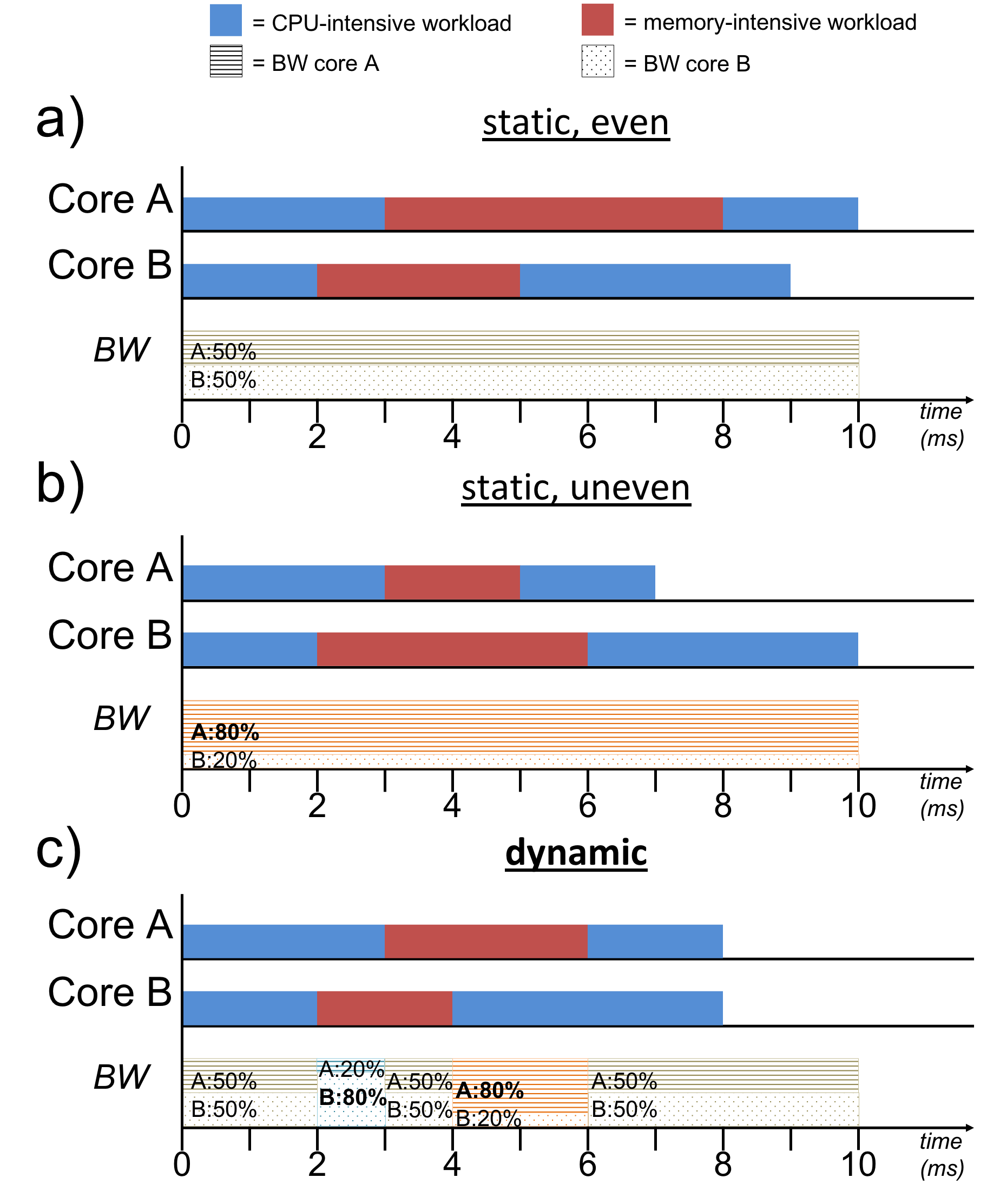}
  \vspace{-0.2cm}
  \caption{Example of static and even (a), static and uneven (b), and
    dynamic (c) memory bandwidth management on a 2-core (A and B)
    system.}
  \label{fig:mem_part}
  \vspace{-0.3cm}
\end{figure}

In both uneven and even partitioning, memory bandwidth allocation is
\emph{statically} decided at design time, i.e., it does not change over
time. The workload on each core, however, can undergo variations in
terms of memory requirements. This is often the case as more/less
memory-intensive tasks (or partitions) are scheduled on each core. As
such, it is natural to consider a scheme where bandwidth-to-cores
assignment is varied over time. In this case, we talk of {\bf dynamic}
bandwidth partitioning, i.e., memory scheduling. Figure~\ref{fig:mem_part} 
depicts one such example. Here, when only one core is executing 
memory-intensive workload, it is given 80\% of the available bandwidth; when 
both are executing the same type of workload, the bandwidth is evenly 
distributed. Under the new
scheme, the system operates at 80\% utilization. In general, dynamic
bandwidth assignment can yield significant performance improvement, because
it is possible to produce an assignment scheme that follows the memory
requirements of scheduled workload over time.

In the next sections, we address the problem of computing the time it
takes in the worst-case to complete execution of workload that: (i)
has known memory and CPU requirements; and (ii) spans over an
arbitrary and known sequence of bandwidth assignments.

\newcommand{\Q}{q}

\section{System Model and Assumptions}
\label{sec:sys_model}

We hereby discuss the assumptions considered in our work. We also
provide the basic terminology and notation required to present our
results. 
\par{\bf Multi-core Model:} in this work, we assume a homogeneous
multi-core system with $m$ cores. We use the index $i$ to refer to any
of the $m$ cores, i.e., $i \in \{1, \ldots, m\}$. We make no assumption
on the cache hierarchy, as we focus on the behavior of tasks with
respect to main memory accesses. We only assume that hits in
last-level cache (LLC) do not generate main memory traffic. Main
memory transactions have fixed size, typically one cache line,
indicated with $L_{size}$. 
We assume that access to main memory is granted to cores/processors
following a round-robin scheme. We assume that the time to perform a
single memory transaction is bounded in the interval: $[L_{min}, L_{max}]$. We 
do not require all memory transactions to be of a fixed size; but assume that, 
in the worst-case, all transactions have the maximum size. As an additional 
simplification, we assume no transaction parallelism, meaning that $L_{max}$ is 
also the maximum amount of
interference that a given core can suffer due to an active memory
transaction directed to a different core. No re-ordering of requests
originated by different cores occurs in the system. This behavior can
be achieved in a traditional COTS DRAM setup by assigning private
memory banks to cores~\cite{palloc, cmu_dram,
  Pellizzoni16:mem_serv}. With private banking, the available DRAM
banks are partitioned among the available $m$ cores. These assumptions
make the considered model compatible with the work in~\cite{SCE_mag,
  SCE_ecrts15, SCE_ecrts17}.

\par{\bf Workload Model:} we consider a partitioned system in which
each task in a set of tasks is statically assigned to one of the $m$
cores~\cite{Davis11:mc_sched_survey}. Since our focus is on the
behavior of workload in memory, we abstract away the details of each
task and only consider the ``load'', or ``workload'' in terms of CPU
time and the number of memory transactions that need to be completed
by a given deadline. The load can correspond to a single task
instance, or to an entire busy-period. This is in line with the
approach followed in~\cite{SCE_ecrts15, SCE_ecrts17,
  gang_mg_analysis}. Reasoning in terms of workload allows us to
remain generic with respect to the exact task scheduling strategy
used at the CPU. For instance, under preemptive rate-monotonic
scheduling (RM), in order to analyze the schedulability of a task
$\tau$, one would consider the deadline-constrained ``load'' comprised
by the execution (CPU time and memory transactions) of one instance of
$\tau$, as well as that of all the instances of interfering
higher-priority tasks. The deadline of the workload will be the
deadline of $\tau$.

Without loss of generality, we model deadline-constrained workload on
a core $i$ under analysis using three parameters: $C_i$, $\mu_i$, and
$D_i$. Here, $C_i$ represents the worst-case amount of time required
for pure execution on the CPU (no memory). For ease of notation, we
will always consider the worst-case execution time in slots of
$L_{max}$ and indicate the latter with $E_i =
\ceil{\frac{C_i}{L_{max}}}$. It must hold that $E_i > 0$. Next,
$\mu_i$ represents the worst-case number of main memory
transactions~\cite{clockdown} to be completed by the relative deadline
$D_i$. We often use $\emusum = E_i + \mu_i$ as a shorthand notation
for the overall CPU and memory requirement of the workload under
analysis. We assume that new workload is always released synchronously
with respect to regulation periods, and scheduling decisions (on both
CPU and memory) are taken at the boundaries of regulation periods.

\par{\bf Memory Bandwidth Regulation Model:} in order to unevenly
partition the memory bandwidth across cores, a budget-based memory
bandwidth regulation scheme is used, such as
MemGuard~\cite{memguard}. In this regulation scheme, per-core
bandwidth regulators use hardware-implemented performance counters to
monitor the number of memory transactions performed by each core over
a period of time $P$. For this reason, $P$ takes the name of
``regulation period''. Note that the number of memory transactions over $P$ is
a measure of bandwidth. Since the maximum latency of a single memory
transaction is $L_{max}$, then in the worst-case it is always possible
to perform $Q = \frac{P}{L_{max}}$ memory transactions in $P$. Each
core can then be assigned a different budget $\Q_i$, as long as $\sum_i
\Q_i \le Q$. However, in order to fully utilize the already constrained
memory bandwidth, we consider the case $\sum_i \Q_i = Q$ without loss
of generality. The budget assigned to all the $m$ cores forms a
vector, namely $\mathcal{Q} = \{\Q_1, \ldots, \Q_m\}$.

The key idea of memory bandwidth regulation is the following. A core
$i$ is given a budget $\Q_i$, which represents the number of memory
transactions that core $i$ is allowed to perform during a regulation
period $P$. The budget is replenished to $\Q_i$ at time zero and at
every instant $k \cdot P$, with $k \in \mathbb{N}$. During a
regulation period, the core executes tasks normally, performing memory
transactions as needed. A hardware performance counter monitors the
number of memory transactions, decreasing the residual budget
accordingly. If core $i$ depletes its budget $\Q_i$ before the next
replenishment, core $i$ is stalled until the next replenishment. $P$
is a system-wide parameter which should be smaller than the minimum
task period in the system. $P$ is often experimentally set to
$1$~ms~\cite{memguard, memguard_tc16}.

\par{\bf Memory Schedule:} in this paper, we assume that the memory
schedule is known, or that a critical instant can be found on the
bandwidth-to-core assignment rule, if an online memory scheduling rule
is used. This opens a whole new set of questions that are out of the
scope of this work: e.g. optimality, or existence of critical instants
for memory schedulers. As depicted in
Figure~\ref{fig:mem_part}(c), a memory schedule $\mathcal{S} = \{B^1,
\ldots, B^N\}$ is a time-ordered sequence of $N$ \emph{memory budget
  assignment intervals} $B^j$. Each $B^j$ is of the form $B^j =
(\mathcal{Q}^j, L^j)$, where $\mathcal{Q}^j = \{\Q^j_1, \ldots, \Q^j_m\}$ is the budget-to-cores
assignment used in interval $j$, and $L^j$ is the length in regulation
periods of interval $j$. For instance, the memory schedule in
Figure~\ref{fig:mem_part}(c) is $\mathcal{S} = \{(\mathcal{Q}, 2),
(\mathcal{Q}', 1), (\mathcal{Q}, 1), (\mathcal{Q}'', 2), (\mathcal{Q},
4)\}$.

\newcommand{\wspan}[2][]{{W^{#1}_{#2}}}

\par{\bf Workload Span:} the goal of this work is to compute the
maximum number of regulation periods required to execute the workload
under analysis to completion. This goes under the name of span, and is
defined below.

\begin{definition}[Span]
  We define \emph{span} as the number of regulation periods to entirely
  complete $E_i$ units of execution and $\mu_i$ memory transactions
  for the considered workload. The span is indicated throughout the
  paper with the symbol $\wspan{i}$.
\end{definition}

The workload, however, may span throughout a number of different memory
scheduling intervals $B^1, \ldots, B^N$. While the total span is
indicated with $\wspan{i}$, the span of the workload over each interval
$B^j$ is indicated with $\wspan[j]{i}$. It must hold that
$\sum_{j=1}^N \wspan[j]{i} = \wspan{i}$. 

Since the intervals have fixed length and are ordered, $\wspan[1]{i}$
must be equal to either $L^1$ or $\wspan{i}$, whichever is
shorter. Then assuming $\wspan{i} > L^1$, the span $\wspan[2]{i}$ over
interval $B^2$ must be equal to the minimum of $L^2$ and $\wspan{i} -
L^1$. In general, noticing that $\sum^{j-1}_{k = 1} L^k$ is the
cumulative length of intervals preceding $B^j$, the execution over
interval $B_j$ must thus be equal to:
\begin{equation} \label{eq:Wsplit}
  \wspan[j]{i} = \max\Big(0, \min\big(L^j, \wspan{i} - \sum^{j-1}_{k =
    1} L^k\big)\Big).
\end{equation}

\section{Memory Stall}
\label{sec:memstall}

A fundamental concept is the notion of \emph{memory stall}. In
general, a memory request originating from the core $i$ under analysis
can be ``stalled'' for two reasons. The first reason is that the
hardware memory arbiter has prioritized one or more other cores over
$i$ for access to the memory subsystem (memory interference). The
second reason is that the core under analysis has exhausted its budget
and is stalled until the beginning of the next regulation period.

\begin{figure*}[!htb]
    \begin{minipage}{0.50\textwidth}
	\centering
	\includegraphics[height=1.7in]{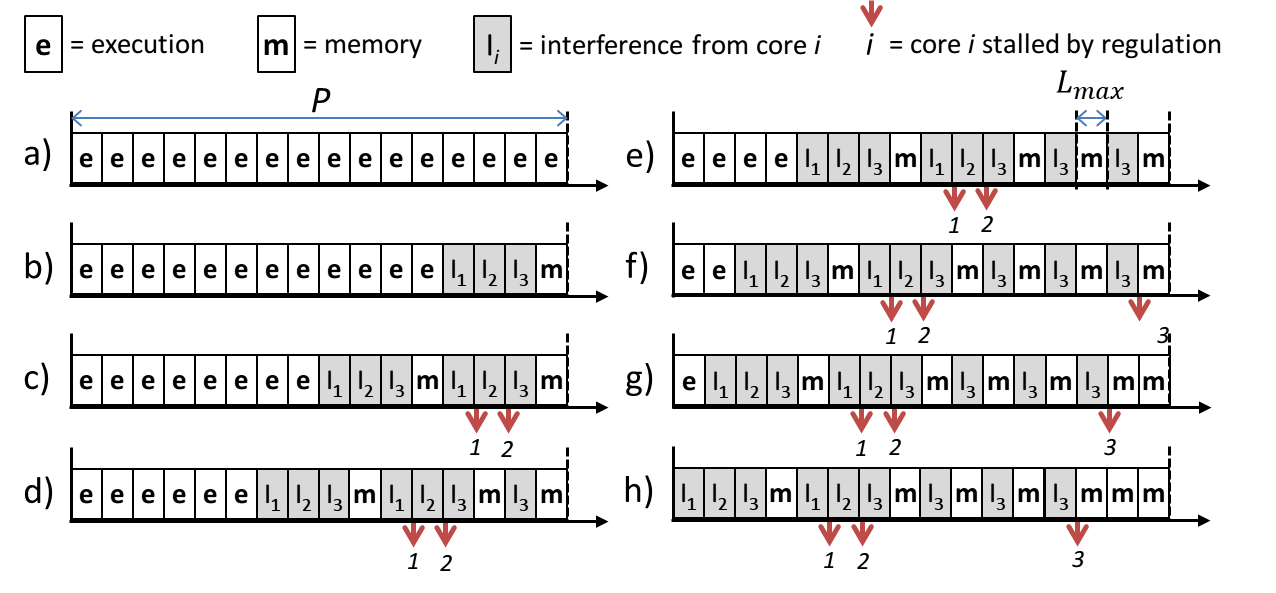}
	\caption{M/C configurations for core $i = 4$ under analysis in
          a 4-core system with $\mathcal{Q}=\{2, 2, 5, 7\}$.}
	\label{fig:mc_config}
    \end{minipage}
    \hspace{0.2in}
    \begin{minipage}{0.45\textwidth}
        \centering
        \includegraphics[height=1.7in]{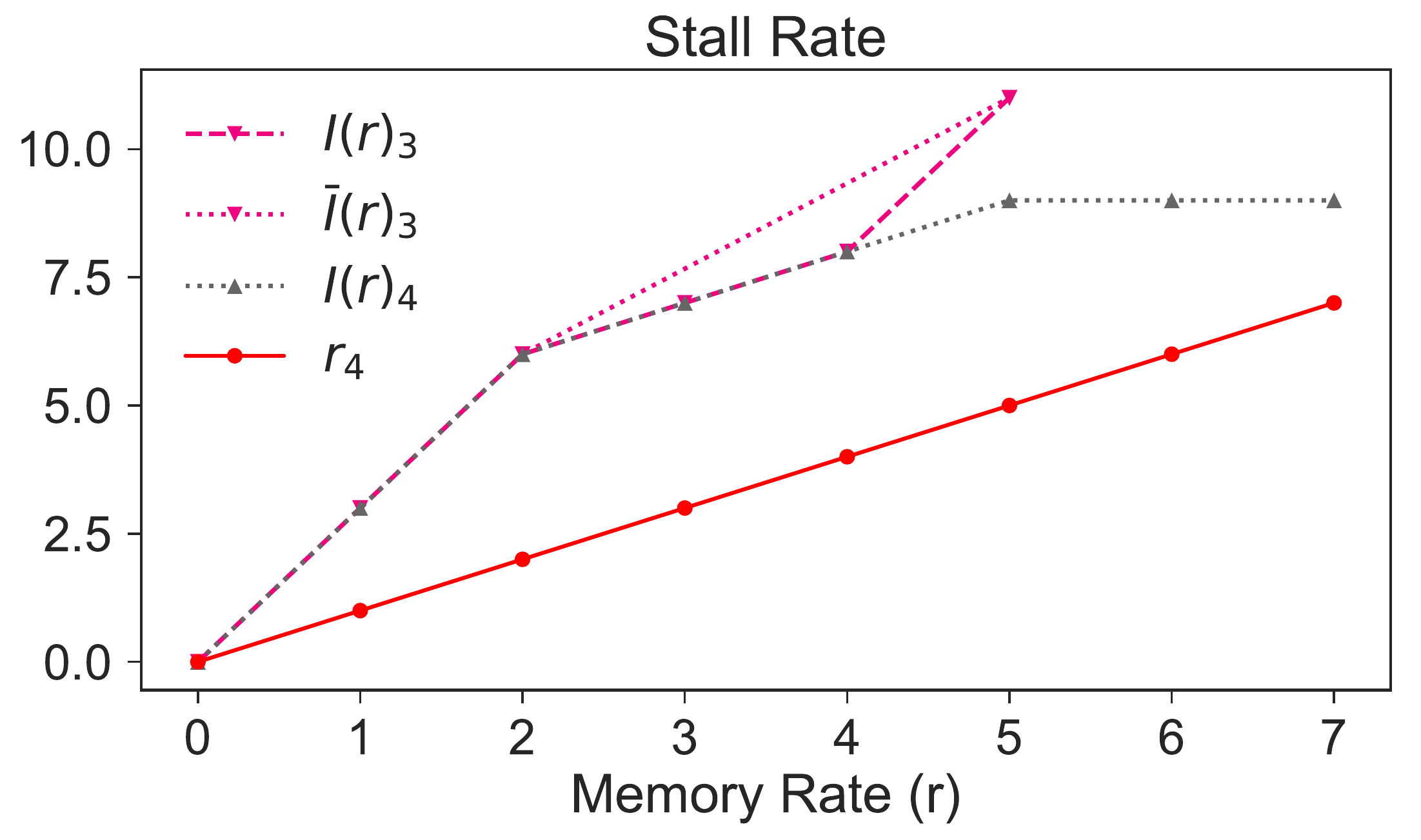}
        \caption{Plot of stall curves for a 4-core system with
          $\mathcal{Q}=\{2, 2, 5, 7\}$.}
        \label{fig:mc_config_stall}
    \end{minipage}
    \vspace{-0.5cm}
\end{figure*}

We make no assumption on the behavior of tasks in cores other than
core $i$. It follows that the maximum stall that can be suffered by a
memory transaction on core $i$ depends only on the number of memory
transactions performed by $i$ in the same regulation period. This is
exemplified in Figure~\ref{fig:mc_config}, where $i = 4$ and
$\mathcal{Q}=\{2, 2, 5, 7\}$. The budget for Core 4 is $\Q_4 = 7$. It
follows that there are 8 possible worst-case scenarios, denoted as
(a)-(h) in the figure. In general, there are always $\Q_i + 1$ possible
cases. In the figure, pure execution is represented with ``e'' and is
considered in slots of length $L_{max}$, as mentioned in
Section~\ref{sec:sys_model}.

The pattern in Figure~\ref{fig:mc_config}(a) depicts the worst-case
memory interference from other cores (cores 1 to 3) provided that core
4 performs zero memory accesses within the regulation period. A more
interesting case is Figure~\ref{fig:mc_config}(e). Here, Core 4
performs 4 memory accesses. For the first two memory accesses, since three cores (1, 2 and 3)
can cause stall, 3 units of stall are accumulated per memory
access. If we depict the stall as a curve, then the ``slope'' of the stall introduced by the first two
accesses is 3. After the first two accesses, cores 1 and 2 are
temporarily stopped due to regulation -- they have exhausted their
respective budgets. Core 3, however, can still cause stall on
transactions from Core 4 under analysis. Hence, the stall slope for the
$3^{rd}$ and $4^{th}$ transactions is 1.

\newcommand{\mrate}[2][]{{r^{#1}_{#2}}}
\newcommand{\muint}[2][]{{\mu^{#1}_{#2}}}

A more efficient way to visualize the possible stall scenarios is by
plotting the per-period memory transactions and resulting stall. The
$\Q_i + 1$ possible cases represent a discrete domain. A corresponding
continuous curve for the memory stall can be derived by ``connecting''
these discrete points. Call $\mrate{i} \in \mathbb{R}_{\ge 0}$ the
number of memory transactions per regulation period being
performed. We introduce the notion of \emph{memory rate}.

\begin{definition}[Memory rate]
  We define as \emph{memory rate} the number of memory transactions
  performed per regulation period. Memory rates are indicated
  throughout the paper with the symbol $\mrate{i}$.
\end{definition}

A memory rate is often used to indicate the rate at which a total
number of memory transactions $\mu_i$ is performed over the span of
the considered workload $\wspan{i}$. Hence, $r_i =
\frac{\mu_i}{\wspan{i}}$. We can also indicate the number of memory
transactions performed during a specific interval $B^j$ as
$\muint[j]{i}$. It must hold that $\sum_{j=1}^N \muint[j]{i} =
\muint{i}$. Analogously, we indicate the memory rate over each
interval as $\mrate[j]{i}$. By definition, we have $\mrate[j]{i} =
\frac{\muint[j]{i}}{\wspan[j]{i}}$.

\vspace{-0.4cm}
\subsection{Memory-stall Curves}
\label{sec:stall_curve}

The memory-stall curve for a core $i$ represents the cumulative
maximum interference-induced stall for a given memory rate $r$ and is
denoted as $I(r)_i$ . Consider same setup used for
Figure~\ref{fig:mc_config}. The memory-stall curves for cores 3 and 4
are provided in Figure~\ref{fig:mc_config_stall}. Considering Core 4,
We have already discussed how the first two transactions introduce
stall at a ``slope'' of 3. This is reflected in the $I(r)_4$ curve,
since the curve has slope 3 when $r \in~]0, 2[$.

For clarity, let us construct the memory-stall curve for core $3$. The
$y$-axis represents the cumulative maximum stall $I(r)_3$ that can be
experienced by workload on Core 3 with a memory rate $r$
($x$-axis). Workload on Core 3 can perform of 0, 1, 2, 3, 4 or 5
memory transactions in a regulation period. The first step is to
compute the maximum stall in each of these cases. If the workload does
not perform any memory transaction ($r = 0$) in a regulation period,
then it will experience no stall, i.e. $I(0)_3 = 0$.  When $r = 1$,
then it can be stalled by a maximum of $1$ memory transaction by each
of the $m-1 = 3$ cores resulting in $I(1)_3 = 3$. Similarly, for all
values of $r$ until $r = \min_i(\Q_1, \ldots, \Q_m)$, the maximum
stall rate $I(r)_3 = (m-1) \cdot r$, hence $I(2)_3 = 3 \cdot 2 =
6$. When Core 3 performs an additional memory transaction, i.e. $r=3$,
it can only be stalled by Core 4, since cores 1 and 2 have been
regulated after their second memory access. Thus, the cumulative stall
rate is $I(3)_3 = I(2)_3 + 1 \cdot 1 = 7$. Similarly, for $r=4$,
$I(4)_3 = I(3)_3 + 1 \cdot 1 = 8$. Finally, for $r=5=\Q_3$, Core 3 is
regulated.  Here the maximum cumulative stall is $I(5)_3 = Q - \Q_3 =
16 - 5 = 11$. The memory-stall curve $I(r)_3$ is obtained by
connecting the discrete values of $I(k)_3, k \in \{0, \ldots, 5\}$
calculated so far.

Generalizing the example provided above, for any fixed budget
$\mathcal{Q}$ we can define the stall curve $I(r)_i$ as follows:
\begin{align} 
  I(r)_i = 
  \begin{cases}
    \sum_{k \neq i} \min(r, \Q_k) &\text{if}~ r < \Q_i\\
    Q - \Q_i &\text{if}~ r = \Q_i
  \end{cases}
  \label{eq:stall_curve_j}
\end{align}

Since the budget assignment $\mathcal{Q}^j$ changes every scheduling
interval $B^j$, a different $I(r)^j_i$ curve needs to be considered on
each interval.

If the resulting curve is concave, then the memory-stall curve is
already final. This is the case for $I(r)_4$ in
Figure~\ref{fig:mc_config_stall}. Conversely, a refinement step is
necessary to produce the final curve. Specifically, we take the
upper-envelope of each of the convex segments to obtain a concave
curve.  The result of this step is depicted as $\Iconc(r)_3$ in
Figure~\ref{fig:mc_config_stall}.

\begin{definition}[Stall rate]
  We define as \emph{stall rate} the amount of memory stall
  $\Iconc(r)_i$ suffered per regulation period with a memory rate
  $r$. When considering multiple intervals, $\Iconc(r)^j_i$ is the
  stall rate for core $i$ on interval $B^j$.
\end{definition}

\newcommand{\mstall}[2][]{{S^{#1}_{#2}}}

If the span over $B^j$ is $\wspan[j]{i}$ and $\muint[j]{i}$
transactions are performed in the interval, we can compute the
worst-case total stall $\mstall[j]{i}$ over $B^j$ as:
\begin{equation}
 \mstall[j]{i} := \Iconc(\mrate[j]{i})_i \cdot \wspan[j]{i} =
 \Iconc\Big(\frac{\muint[j]{i}}{\wspan[j]{i}}\Big)_i \cdot \wspan[j]{i}.
 \label{eq:mstall}
\end{equation} 

It follows that the total stall is $\mstall{i} = \sum_{j=1}^N
\mstall[j]{i}$. If the maximum memory stall that can be suffered by
the workload under analysis can be derived, then the worst-case amount
of time (in multiples of $L_{max}$) required to complete the
considered workload is $\wspan{i} \cdot Q = \beta_i + \mstall{i}$. The
rest of the paper is concerned with the calculation of the maximum
total stall, and hence span, over a generic memory schedule
$\mathcal{S} = \{B^1, \ldots, B^N\}$.

For a fixed budget $\mathcal{Q}$, a given memory rate $\mrate{i}$ and
span $\wspan{i}$, Lemma~\ref{lem:stall_ub} guarantees that computing
$\mstall{i}$ according to Equation~\ref{eq:mstall} always results in
an upper-bound on the maximum possible memory stall.

\begin{lemma}
  \label{lem:stall_ub}
  $\mstall{i} = \Iconc(\muint{i}/\wspan{i})_i \cdot \wspan{i}$ is an
  upper bound to the cumulative stall suffered by a workload on core
  $i$ that performs $\muint{i}$ memory accesses over $\wspan{i}$
  regulation periods, with $\muint{i}~\leq~\wspan{i}~\cdot \Q_i$.
 \end{lemma}
  
\begin{IEEEproof}
    In each of the $\wspan{i}$ regulation periods, a number of memory
    accesses between $0$ and $\Q_i$ could have been performed; hence,
    note that we cannot have $\muint{i} > \wspan{i} \cdot \Q_i$. Let us
    indicate with $a_k$ the number of periods in which $k$ memory
    accesses were performed. It must hold that $\sum^{\Q_i}_{k=0} a_k =
    \wspan{i}$. We can then write:
    \begin{equation}
      \muint{i} = a_0 \cdot 0 + a_1 \cdot 1 + \ldots + a_{\Q_i} \cdot \Q_i.
    \end{equation}    
    The cumulative stall suffered over $\wspan{i}$ can be computed as:
    \begin{equation}
      \label{eq:stall-rate}
      a_0 \cdot I(0)_i + a_1 \cdot I(1)_i + \ldots + a_{\Q_i} \cdot I(\Q_i)= \sum^{\Q_i}_{k=0} I(k)_i \cdot a_k.
    \end{equation}
    Consider now computing the stall rate as $\Iconc(\muint{i}/\wspan{i})_i$. From
    Equation~\ref{eq:stall-rate}, by $I(r) \le \Iconc(r)$ we have:
    \begin{equation}
     \sum^{\Q_i}_{k=0} I(k)_i \cdot a_k \le \sum^{\Q_i}_{k=0} \Iconc(k)_i \cdot a_k
    \end{equation}
    Next recall that by definition of concavity for a generic function $f(x)$, it must hold that:    
    \begin{equation}
      \lambda_k \in \mathbb{R}~~\text{s.t.}~\sum_k \lambda_k = 1 \implies f\Big(\sum_k x_k \lambda_k\Big) \ge \sum_k f(x_k)\lambda_k
    \end{equation}
    Note that:
    \begin{equation}
      \frac{\sum^{\Q_i}_{k=0} a_k}{\wspan{i}} = \sum^{\Q_i}_{k=0} \frac{a_k}{\sum^{\Q_i}_{k=0} a_k} = 1.
    \end{equation}
    Hence, we can write:
    \begin{equation}
      \footnotesize
      \frac{\sum^{\Q_i}_{k=0} \Iconc(k)_i \cdot a_k}{\wspan{i}} = \sum^{\Q_i}_{k=0}\frac{ \Iconc(k)_i \cdot a_k}{\sum^{\Q_i}_{k=0} a_k} 
      \le \Iconc\Bigg(\sum^{\Q_i}_{k=0}\frac{k \cdot a_k}{\sum^{\Q_i}_{k=0} a_k}\Bigg)_i = \Iconc\Big(\frac{\muint{i}}{\wspan{i}}\Big)_i.
    \end{equation}
This implies that $\Iconc(\muint{i} / \wspan{i})_i \cdot \wspan{i}$ is
an upper bound to the cumulative stall $\sum^{\Q_i}_{k=0} I(k)_i \cdot
a_k$ suffered by the workload for \emph{any} pattern of memory
accesses over $\wspan{i}$ periods, concluding the proof.
\end{IEEEproof}

\section{WCET under Static Memory Budget}
\label{sec:wcet_1_bud}

In this section, we present a fixed-point iterative algorithm to compute the
worst-case length of the workload on a core under analysis $i$ under
static memory budget $\mathcal{Q}$. This is useful to understand the
basic mechanisms to compute the span over a generic single memory
scheduling interval. 

In each iteration, the algorithm recomputes the maximum stall and thereby, the 
workload span, based on the workload span from the previous iteration (except 
for the base iteration) and the corresponding memory schedule. The key 
intuition behind iterative recomputation is that the increase in workload span 
in an iteration is likely to increase the maximum stall in the consecutive 
iteration due to a different worst-case distribution of memory requests across 
\textit{(a)} different per memory-stall curves and/or \textit{(b)} different 
memory scheduling intervals.

In the rest of the paper, we will always focus on the generic core
under analysis. As such, we will drop the index $i$ from all the
notation introduced so far, unless required to resolve an
ambiguity. Since we will be introducing a series of iterations of the algorithm, we subscript the iteration number (e.g., $(k)$) in the
notation introduced so far.
	
\par{\bf Iterative Algorithm:}
\label{sec:wcet_1_bud_alg}
the span $W$ over a static memory budget can be computed using
Equation~\ref{eq:iteration_static}.
\begin{align}
\wspan{(0)} & \gets \lceil \emusum / Q \rceil, \nonumber \\
\wspan{(k)} & \gets \Big\lceil \big( \emusum + \Iconc(\min(\muint{} / \wspan{(k-1)}), \Q) \cdot \wspan{(k-1)} \big) / Q \Big\rceil, \label{eq:iteration_static}
\end{align}
where the iteration continues until convergence with $\wspan{(k)} =
\wspan{(k-1)}$, or until $\wspan{(k)} \cdot Q \cdot L_{max} > D$. In
the latter case, the workload in not schedulable. Since $\Iconc(r)$ is
only defined for $r \in [0, q]$, the term $\Iconc(\min(\muint{} /
\wspan{(k-1)}), \Q)$ ensures that the function is never evaluated on a
value outside its domain.

\begin{theorem}\label{theorem1}
  \label{thm:wcet_single}
  The iteration in Equation~\ref{eq:iteration_static} terminates in a
  finite number of steps by either obtaining a value $\wspan{(k)}
  \cdot Q \cdot L_{max} > D$, or by converging, in which case
  $\wspan{(k)}$ is an upper bound on the span of the workload on the
  core under analysis.
\end{theorem}

\begin{IEEEproof}[Proof Sketch]  Notice that we omit the proof for Theorem~\ref{theorem1} here, as it is 
  a corollary of the more general Theorem~\ref{thm:length_dynamic}. As such, 
  the proof is provided in the Appendix.
\end{IEEEproof}

For ease of explanation, Section~\ref{sec:wcet_1_bud_illus}
illustrates how to apply the algorithm in a specific
instance. Subsequently, Section~\ref{sec:wcet_n_bud} presents the
generic algorithm.

\par{\bf Example of WCET over Static Budget:}
\label{sec:wcet_1_bud_illus}
consider the static budget $\mathcal{Q}=\{2,2,5,7\}$. Let us now
compute the span $W$ of the workload with $E=40$ and $\muint{}=35$
(i.e. $\beta = 75$) executing on Core $3$. For simplicity, we ignore
the workload's deadline $D$ and focus only on its length. Since
workload on Core $3$ is being analyzed, we consider the stall curve
$\Iconc(r)_3$ in Figure~\ref{fig:mc_config_stall}, with $Q = 16$ and
$q = 5$.

The first step in the iterative Equation~\ref{eq:iteration_static} is
$\wspan{(0)} = \ceil{\emusum/Q} = \ceil{75/16} = 5$. We then have:
\begin{align*}
  \wspan{(1)} &= \ceil{(75 + \Iconc(\min(35/5, 5)\cdot 5)/16} \\
  &= \ceil{(75 + \Iconc(5)\cdot 5)/16} = \ceil{(75 + 11 \cdot 5)/16} = 9
\end{align*}
Since $\wspan{0} \neq \wspan{k}$, an additional iteration needs to be
performed. We have:
\begin{align*}
  \wspan{(2)} &= \ceil{(75 + \Iconc(\min(35/9, 5) \cdot 9)/16} \\
  &= \ceil{(75 + \Iconc(3.88)\cdot 9)/16} = \ceil{(75 + 9 \cdot 9)/16} = 10
\end{align*}
No convergence has been reached yet, so one more iteration is
performed:
\begin{align*}
  \wspan{(3)} &= \ceil{(75 + \Iconc(\min(35/10, 5) \cdot 10)/16} \\
  &= \ceil{(75 + \Iconc(3.5)\cdot 10)/16} = \ceil{(75 + 8.5 \cdot 10)/16} = 10
\end{align*}
Since $\wspan{(3)} = \wspan{(2)}$, convergence has been reached and
the worst-case length (in multiples of $L_{max}$) for the workload
under analysis can be computed as $\wspan{(3)} \cdot Q = 160$.

\section{WCET under Dynamic Memory Budget}
\label{sec:wcet_n_bud}

\begin{figure*}
  \centering \includegraphics[width=0.9\textwidth]{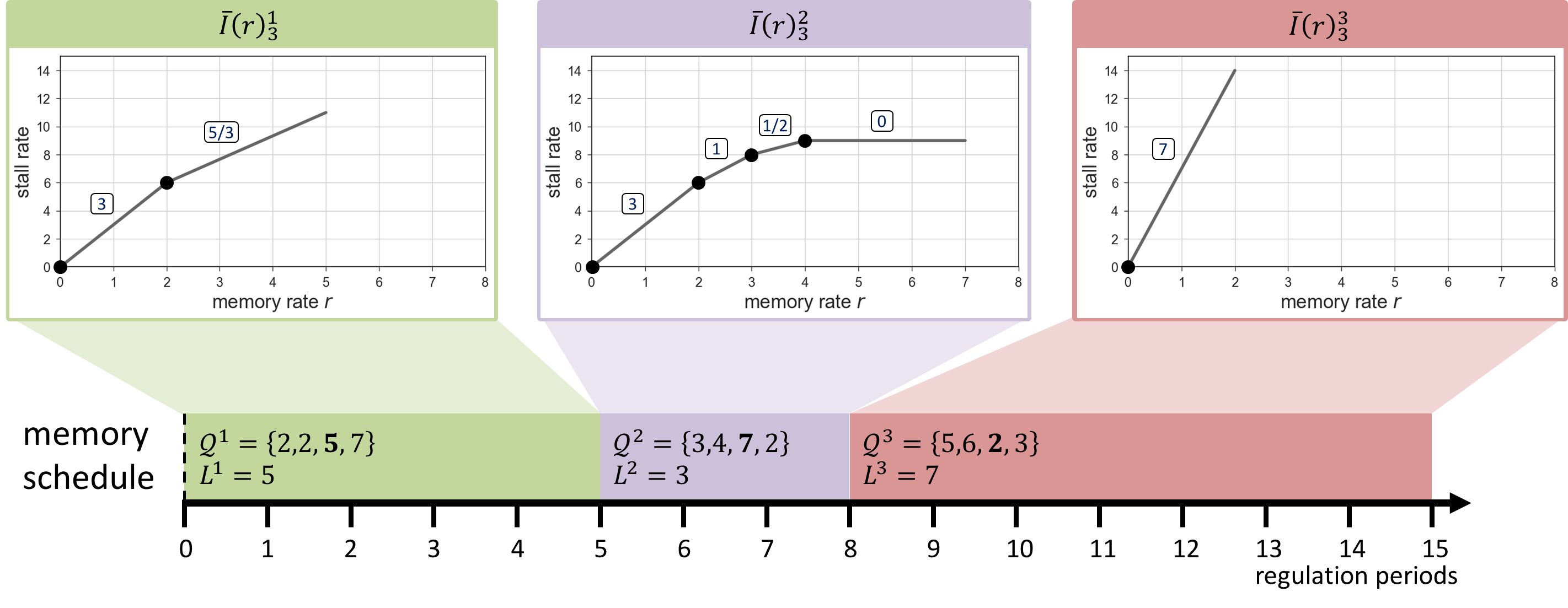}
  \vspace{-0.3cm}
  \caption{Example of memory schedule of length $15$ regulation
    periods composed of 3 intervals $B^1, B^2, B^3$ of length 5, 3,
    and 7, respectively. Considering Core 3, the curves
    $\Iconc(r)^1_3, \Iconc(r)^2_3, \Iconc(r)^3_3$ are reported above
    each interval. For each $\Iconc(r)^j_3$ curve, segment start
    points are highlights with a $\bullet$ and slopes are annotated
    above the segments.}
  \label{fig:mult_budget}
  \vspace{-0.5cm}
\end{figure*}

In this section, we extend our analysis to the case of dynamic
bandwidth assignment. In this case, the workload could span across one
or more memory scheduling intervals $B^1, \dots, B^N$. Recall from
Section~\ref{sec:background_gen} that each interval $B^j =
(\mathcal{Q}^j, L^j)$ is characterized by a budget-to-cores assignment
$\mathcal{Q}^j = \{\Q^j_1, \ldots, \Q^j_m\}$ and a length $L^j$
expressed in number of regulation periods. For the core under
analysis, it is always possible to compute $\Iconc(r)^j$ using
Equation~\ref{eq:stall_curve_j}, i.e. the stall curve resulting from
$\mathcal{Q}^j$. 

Similarly to Equation~\ref{eq:iteration_static}, we follow an
iterative approach. Let once again $\wspan{(k)}$ be the span of the
considered workload at iteration $k$. It is always possible to
determine the span of the workload $\wspan[j]{(k)}$ in each of the
intervals $B^j$ using Equation~\ref{eq:Wsplit}.

\textbf{Example:} To better understand this setup, consider the
situation depicted in Figure~\ref{fig:mult_budget}. In this case, the
memory schedule is composed of 3 intervals $B^1, B^2, B^3$. The
intervals have length $L^1 = 5, L^2 = 3$, and $L^3 = 7$,
respectively. Similarly we have three budget-to-cores assignments
$\mathcal{Q}^1, \mathcal{Q}^2, \mathcal{Q}^3$ with the budgets values
specified in the figure. Assume that we want to analyze the behavior
of workload on Core $i = 3$: the figure reports the $\Iconc(r)^1_3,
\Iconc(r)^2_3, \Iconc(r)^3_3$ curves that need to be considered in
each of the intervals. Let us assume that at a certain iteration $k$,
we have $\wspan{(k)} = 9$ and apply Equation~\ref{eq:Wsplit}. We
obtain $\wspan[1]{(k)} = \max\big(0, \min(5,9) \big) = 5,
\wspan[2]{(k)} = \max\big(0, \min(3,9 - 5) \big) = 3$, and
$\wspan[3]{(k)} = \max\big(0, \min(7,9 - 8 ) \big) = 1$. Note that it
always hold $\sum_{j=1}^N \wspan[j]{(k)} = \wspan{(k)}$.

Assume that the span $\wspan{(k)}$ and hence the various
$\wspan[j]{(k)}$ terms at a given iteration $k$ are known. Then, the
challenge is to determine how to distribute the total $\muint{}$
memory transactions among the $B^1, \ldots, B^N$ intervals in a way
that maximizes the overall stall. A \emph{distribution} of memory
transactions simply means that we derive the quantities
$\muint[j]{(k)}$ for each $B^j$ interval. Obviously, it must hold that
$\sum_{j=1}^N \muint[j]{(k)} = \muint{}$. But among all the possible,
valid distributions, we are interested in the one\footnote{This
  distribution may not be unique.} that maximizes the overall stall,
i.e. $\mstall{(k)} = \sum_{j=1}^N \mstall[j]{(k)}$, where each of the
$\mstall[j]{(k)}$ terms is defined as in Equation~\ref{eq:mstall}.

To simplify exposition, we first introduce an optimization problem in
Algorithm~\ref{alg:opt} that computes the number $\muint[j]{(k)}$ of
memory requests assigned to each interval $B^j$ at the $k$-th
iteration to maximize the overall stall $\mstall{(k)}$. Then, in
Section~\ref{sec:algorithm} we show how to efficiently solve the
optimization problem. Note that once $\muint[j]{(k)}$ has been determined,
based on Lemma~\ref{lem:stall_ub} the stall in interval $B^j$ can be
upper-bounded as $\mstall[j]{(k)} = \Iconc(\muint[j]{(k)} /
\wspan[j]{(k)})^j \cdot \wspan[j]{(k)}$ . Hence, the cumulative stall
at Line~\ref{alg:opt_objective} of the algorithm is computed as
$\mstall{(k)} = \sum_{j=1}^N \mstall[j]{(k)}$.  Due to regulation, in
Line~\ref{alg:WQ} we assign at most $\Q^j$ memory transactions in any
regulation period inside interval $B^j$. This is equivalent to the
following constraint: the number of transactions $\muint[j]{(k)}$
performed in $B^j$ is upper bounded by $\wspan[j]{(k)} \cdot
\Q^j$. Finally, since we know that the workload comprises at most
$\mu$ requests, it must hold $\sum_{j =1}^N \muint[j]{(k)} \leq \mu$
at Line~\ref{alg:split_memory}.

Based on Algorithm~\ref{alg:opt},
$\wspan[j]{(k)}$ is then computed according to the following
iteration:
\begin{align}
\wspan{(0)}& \gets \lceil \emusum / Q \rceil, \nonumber \\
\wspan{(k)} & \gets \Big\lceil \big( \emusum + \sum_{j=1}^N \Iconc(\muint[j]{(k-1)} / \wspan[j]{(k-1)})^j  \cdot \wspan[j]{(k-1)} \big) / Q \Big\rceil. \label{eq:iteration_dynamic}
\end{align}

Note that at each iteration $k > 0$, the values $\wspan[j]{(k-1)}$ are
computed using Equation~\ref{eq:Wsplit} from $\wspan{(k-1)}$, and the
values $\muint[j]{(k-1)}$ are computed using
Algorithm~\ref{alg:opt}. As in Section~\ref{sec:wcet_1_bud}, the
iteration continues until convergence or $\wspan{(k)} \cdot Q \cdot
L_{max} > D$.

\lstset{
  morecomment=[s][\color{red}\scriptsize\itshape]{/*}{*/},
  mathescape=true,
    columns=fullflexible,
  keepspaces=true,
  xleftmargin=0.1\linewidth,
  xrightmargin=0.05\linewidth,
  float=tp,
  floatplacement=tbp,
}

\begin{framed}
\begin{lstlisting}[caption={Stall maximization over multiple intervals}, 
label=alg:opt, abovecaptionskip=-\medskipamount, belowcaptionskip=0.2cm, 
mathescape=true, numbers=left, 
numberstyle=\scriptsize,basicstyle=\footnotesize,escapeinside={(*@}{@*)}]
Input: $B^1, \ldots, B^N$ (*@\hfill@*)/* sequence of intervals */
Input: $\wspan[1]{(k)}, \ldots, \wspan[N]{(k)}$ (*@\hspace*{\fill}@*)/* span in each interval */
Input: $\mu$ (*@\hfill@*)/* total memory requests */

Output: $\muint[1]{(k)}, \ldots, \muint[N]{(k)}$ (*@\hfill@*)/* memory requests in each interval */

Maximize: (*@\label{alg:opt_objective_capt}@*) (*@\hfill@*)/* max cumulative stall */
   $\mstall{(k)} = \sum_{j=1}^N \mstall[j]{(k)} = \sum_{j=1}^N 
   \Iconc(\muint[j]{(k)} / \wspan[j]{(k)})^j  \cdot \wspan[j]{(k)}$ (*@\label{alg:opt_objective}@*)
   
Subject to:
   $\muint[j]{(k)} \in \mathbb{N}$ (*@\label{alg:natural}@*) (*@\hfill@*)/*    number of requests is natural */
   $\muint[j]{(k)} \leq \wspan[j]{(k)} \cdot \Q^j$ (*@\label{alg:WQ}@*) (*@\hfill@*)/* max $\Q^j$ transactions per regulation period */
   $\sum_{j =1}^N \muint[j]{(k)} \leq \mu$ (*@\label{alg:split_memory}@*) (*@\hfill@*)/* total requests constraint */
\end{lstlisting}
\end{framed}

\textbf{Example:} Suppose we are analyzing the behavior of workload
with $E = 15$ and $\muint{} = 25$ on Core 3 under the memory schedule
depicted in Figure~\ref{fig:mult_budget}. Assume that at a given step
$k$ we have $\wspan{(k)} = 6$. In this case we have $\wspan[1]{(k)} =
5, \wspan[2]{(k)} = 1$ and $\wspan[3]{(k)} = 0$. Invoking
Algorithm~\ref{alg:opt} returns the memory-to-intervals distribution
that maximizes the overall stall, in this case: $\muint[1]{(k)} = 22,
\muint[2]{(k)} = 3, \muint[3]{(k)} = 0$. The stall per interval can be
computed as $\mstall[j]{(k)} =
\Iconc\big({\muint[j]{(k)}/\wspan[j]{(k)}}\big) \cdot
\wspan[j]{(k)}$. For this example, we have $\mstall[1]{(k)} = 50,
\mstall[2]{(k)} = 8$ and $\mstall[3]{(k)} = 0$. The new value of
$\wspan{(k+1)}$ can then be computed as $\wspan{(k+1)} = \ceil{(35 +
  50 + 8)/16} = 6$. Note that in this case
Equation~\ref{eq:iteration_dynamic} has reached convergence.

\subsection{Proof of Correctness} \label{sec:dynamic_proofs}

We now formally prove that Equation~\ref{eq:iteration_dynamic}
computes a valid upper bound for the workload length in number of
regulation periods. We begin with some helper lemmas;
Lemma~\ref{lm:dynamic_monotonic} show that the value of $\wspan{(k)}$
increases monotonically, which is required for the iteration to
terminate, while Lemma~\ref{lm:equality} shows that if the iteration
converges, we are able to distribute all $\mu$ memory requests among
the $N$ memory scheduling intervals.

\begin{lemma}\label{lm:dynamic_monotonic}
  At each iteration step in Equation~\ref{eq:iteration_dynamic} it
  holds: $\wspan{(k)} \geq \wspan{(k-1)} > 0$.
\end{lemma}
\begin{IEEEproof}
  First note that functions $\Iconc(r)^j$ are concave and $\Iconc(0)^j
  = 0$. For any such function and positive constant $\muint{}$, one can
  prove that $\Iconc(\muint{}/x)^j \cdot x$ is monotonic non-decreasing in
  $x > 0$ (a formal proof is reported in Lemma~\ref{lm:monotonicI} in
  Appendix). The proof then proceeds by induction over the index $k$.

  \textbf{Base Case:} Since we assume $\beta > 0$, we have
  $\wspan{(0)}> 0$. Furthermore, since by definition all
  $\wspan[j]{(0)}$ terms are non-negative and functions $\Iconc(r)^j$
  have non-negative ranges, $\sum_{j=1}^N \Iconc(\muint[j]{(0)} /
  \wspan[j]{(0)})^j \cdot \wspan[j]{(0)}$ is non-negative. By
  definition of Equation~\ref{eq:iteration_dynamic}, this implies
  $\wspan{(1)} \geq \wspan{(0)}> 0$.

  \textbf{Inductive Step:} Consider $k \geq 2$. Note that in
  Equation~\ref{eq:iteration_dynamic}, $\wspan{(k)}$ is computed based
  on the value of $\wspan{(k-1)}$, from which we obtain the values of
  $\wspan[j]{(k-1)}$ in Equation~\ref{eq:Wsplit} and
  $\muint[j]{(k-1)}$ in Algorithm~\ref{alg:opt}; similarly,
  $\wspan{(k-1)}$ is computed based on $\wspan{(k-2)}$ and
  $\wspan[j]{(k-2)}, \muint[j]{(k-2)}$. By induction hypothesis, we
  have $\wspan{(k-1)} \geq \wspan{(k-2)} > 0$; based on
  Equation~\ref{eq:Wsplit}, this implies $\wspan[j]{(k-1)} \geq
  \wspan[j]{(k-2)}$ for all intervals $B^j$.

  Now consider Line~\ref{alg:opt_objective} of
  Algorithm~\ref{alg:opt}: since $\Iconc(\mu/x)^j \cdot x$ is
  monotonic for $x > 0$, it must hold:
  \begin{align}
    & \sum_{j=1}^N \Iconc(\muint[j]{(k-2)} / \wspan[j]{(k-1)})^j  \cdot \wspan[j]{(k-1)} \geq \nonumber \\
    & \sum_{j=1}^N \Iconc(\muint[j]{(k-2)} / \wspan[j]{(k-2)})^j  \cdot \wspan[j]{(k-2)}.
  \end{align}
  In other words, when running Algorithm~\ref{alg:opt} at iteration
  $k$ based on the values $\wspan[j]{(k-1)}$, there exists an
  assignment of variables ($\muint[j]{(k-1)} = \muint[j]{(k-2)}$,
  i.e., the same assignment as the previous iteration) that results in
  a value of the objective function that is greater than or equal to
  the one at iteration $k-1$. Furthermore, the assignment
  $\muint[j]{(k-1)} = \muint[j]{(k-2)}$ is feasible, in the sense that
  it satisfies the constraints at
  Lines~\ref{alg:natural}-\ref{alg:split_memory} of the algorithm:
  note that $\muint[j]{(k-2)} \leq \wspan[j]{(k-2)} \cdot \Q^j$
  implies $\muint[j]{(k-1)} \leq \wspan[j]{(k-1)} \cdot \Q^j$ since
  $\wspan[j]{(k-1)} \geq \wspan[j]{(k-2)}$.  Hence, given that the
  optimization problem is maximizing the objective function, it is
  guaranteed to find an assignment for variables $\muint[j]{(k-1)}$
  such that:
  \begin{align}
    & \sum_{j=1}^N \Iconc(\muint[j]{(k-1)} / \wspan[j]{(k-1)})^j  \cdot \wspan[j]{(k-1)} \geq \nonumber \\
    & \sum_{j=1}^N \Iconc(\muint[j]{(k-2)} / \wspan[j]{(k-2)})^j  \cdot \wspan[j]{(k-2)}.
  \end{align}
  In turn by definition of Equation~\ref{eq:iteration_dynamic} this
  implies $\wspan{(k)} \geq \wspan{(k-1)}$, concluding the induction
  step.
\end{IEEEproof}

\begin{lemma}\label{lm:equality}
  If the iteration in Equation~\ref{eq:iteration_dynamic} converges to
  a value $\wspan{(k)}$, then there exists a feasible assignment to
  variables $\muint[j]{(k)}$ that maximizes the objective function at
  Line~\ref{alg:opt_objective} of Algorithm~\ref{alg:opt} and for
  which $\sum_{j=1}^N \muint[j]{(k)} = \mu$.
\end{lemma}
\begin{IEEEproof}
  Note that for a feasible assignment it cannot hold $\sum_{j=1}^N
  \muint[j]{(k)} > \mu$ due to the constraint at
  Line~\ref{alg:split_memory}. Hence by contradiction, assume that for
  all assignments that maximize the objective function it holds:
  $\sum_{j=1}^N \muint[j]{(k)} < \mu$.

  By definition, all $\wspan[j]{(k)}$ terms are non
  negative. Furthermore, functions $\Iconc(r)^j$ have non-negative
  ranges. Hence, increasing the value of a variable $\muint[j]{(k)}$
  cannot cause the objective function to decrease. Therefore,
  $\sum_{j=1}^N \muint[j]{(k)} < \mu$ must hold even when each
  variable $\muint[j]{(k)}$ is assigned its maximum value, which is
  $\muint[j]{(k)} = \wspan[j]{(k)} \cdot \Q^j$ based on the constraint
  at Line~\ref{alg:WQ}. We thus obtain: $\sum_{j=1}^N \wspan[j]{(k)}
  \cdot \Q^j < \mu$.  Furthermore, note that we have
  $\Iconc(\muint[j]{(k)} / \wspan[j]{(k)})^j = \Iconc(\Q^j)^j = Q -
  \Q^j$.

  Finally, given $E > 0$ and based on
  Equation~\ref{eq:iteration_dynamic} at convergence, we derive:
  \begin{align}
    \wspan{(k)} & = \Big\lceil \big( \emusum + \sum_{j=1}^N \Iconc(\muint[j]{(k)} / \wspan[j]{(k)})^j  \cdot \wspan[j]{(k)} \big) / Q \Big\rceil \nonumber \\
    & = \Big\lceil \big( E + \mu - \sum_{j=1}^N \Q^j \cdot \wspan[j]{(k)} + \sum_{j=1}^N Q \cdot \wspan[j]{(k)}  \big) / Q \Big\rceil \nonumber \\
    & > \Big\lceil \big( E +  \sum_{j=1}^N (Q \cdot \wspan[j]{(k)}) \big) / Q \Big\rceil \nonumber \\
    & = \lceil E / Q + \wspan{(k)} \rceil \nonumber \\
    & \geq 1+ \wspan{(k)},
  \end{align}
  which is a contradiction.
\end{IEEEproof}

\begin{theorem}
  \label{thm:length_dynamic}
  The iteration in Equation~\ref{eq:iteration_dynamic} terminates in a
  finite number of steps by either obtaining a value $\wspan{(k)}
  \cdot Q \cdot L_{max} > D$ or converging, in which case
  $\wspan{(k)}$ is an upper bound to worst-case span of the workload
  on the core under analysis.
\end{theorem}
\begin{IEEEproof}
   We first show that the algorithm terminates. By
   Lemma~\ref{lm:dynamic_monotonic}, $\wspan{(k)} \geq
   \wspan{(k-1)}$. Since $\wspan{(k)}$ is a natural number, it follows
   that the algorithm must either converge or terminate with a value
   of $\wspan{(k)}$ greater than the deadline in a finite number of
   steps.
    
   Hence, assume that the algorithm converges to $\wspan{(k)}$. Based
   on Lemma~\ref{lm:equality}, we can find an assignment to variables
   $\muint[j]{(k)}$ that maximizes the stall in the objective function
   of Algorithm~\ref{alg:opt} and such that $\sum_{j=1}^N
   \muint[j]{(k)} = \mu$. Hence, the assignment is valid, in the sense
   that the workload is able to perform its worst-case number of
   memory transactions. Furthermore, due to Line~\ref{alg:WQ}, it
   holds $\muint[j]{(k)} \leq \wspan[j]{(k)} \cdot \Q^j$ for all
   intervals; hence, by Lemma~\ref{lem:stall_ub} and for each $B^j$,
   $\mstall[j]{(k)} = \Iconc(\muint[j]{(k)} / \wspan[j]{(k)})^j \cdot
   \wspan[j]{(k)}$ is an upper bound to the stall when performing
   $\muint[j]{(k)}$ memory accesses in $\wspan[j]{(k)}$ regulation
   periods. Now given that Algorithm~\ref{alg:opt} maximizes the
   objective function at Line~\ref{alg:opt_objective} over all
   possible assignments to variables $\muint[j]{(k)}$, it follows that
   $\sum_{j=1}^N \mstall[j]{(k)} = \mstall{(k)}$ is an upper bound to
   the cumulative stall when performing $\muint{}$ memory accesses
   over $\sum_{j=1}^N \wspan[j]{(k)} = \wspan{(k)}$ regulation
   periods.
   
   Finally, by definition, the worst-case length of the workload can be
   obtained (in number of slots) as the sum of $\beta$ and the stall
   suffered by the workload. By convergence to $\wspan{(k)}$, we have:
   \begin{equation}
     \wspan{(k)} \cdot Q = \ceil[\Big]{\frac{\wspan{(k)}}{Q}} \cdot Q
     \geq \emusum + \sum_{j=1}^N \mstall[j]{(k)} = \emusum + \mstall{(k)},
   \end{equation}
   and since $\mstall{(k)}$ is an upper bound to the stall suffered in
   $\wspan{(k)}$ regulation periods, this implies that $\wspan{(k)}$
   is indeed an upper bound to the total span of the workload.
\end{IEEEproof}

\subsection{Implementing the Stall Algorithm} \label{sec:algorithm}

In this section, we show how to efficiently implement
Algorithm~\ref{alg:opt}. The algorithm is similar to a concave
optimization problem, except that variables are integer rather than
real.

\newcommand{\epts}[1]{{\mathcal{E}^{#1}}}
\newcommand{\nextfn}[1]{{\textrm{next}^{#1}}}
\newcommand{\slopefn}[1]{{\textrm{slope}^{#1}}}

By construction, each $\Iconc(r)^j$ function is a concave piece-wise
linear, and can be thought as a sequence of \emph{segments} with
decreasing slope. For each $\Iconc(r)^j$ curve, consider the set of
integer values of $r$ corresponding to the beginning of a
segment. Call each of this values a \emph{start point}, and $\epts{j}$
the set of all the start points in $\Iconc(r)^j$. Start points are
highlighted with a solid dot ($\bullet$) in
Figure~\ref{fig:mult_budget}. Considering Core 3, for the example in
the figure we have: $\epts{1} = \{0, 2\}, \epts{2} = \{0, 2, 3, 4\}$,
and $\epts{3} = \{0\}$.

Using this formulation, we introduce two helper functions defined on
$\epts{j}$ for $\Iconc(r)^j$. First, the $\nextfn{j}(r)$ function
returns the next start point strictly greater $r$:
\begin{equation}
 \nextfn{j}(r) := \min_p \{p \in \epts{j} ~|~ p > r\}.
\end{equation}
Second, the function $\slopefn{j}(r)$ simply returns the slope of the
segment at $r$. If $r$ is a start point, the function returns the
slope of the starting segment. Formally:
\begin{equation}
 \slopefn{j}(r) := \Big( \Iconc\big( \textrm{next}^j(r) \big)^j
- \Iconc(r)^j \Big) / \big( \textrm{next}^j(r) - r \big)
\end{equation}
All the slopes are annotated in Figure~\ref{fig:mult_budget} right
above the corresponding segment.

Algorithm~\ref{alg:concave} first initializes all variables
$\muint[j]{(k)}$ to zero. Then, the algorithm iterates over
Lines~\ref{alg:do}-\ref{alg:until} until either (1)
$\sum_{j=1}^N \muint[j]{(k)} = \mu$ holds, meaning that all $\mu$
memory transactions have been distributed among the $N$ memory
scheduling interval; or (2) $\muint[j]{(k)}
= \wspan[j]{(k)} \cdot \Q^j$ for all intervals, meaning that we cannot
assign any more memory transactions due to the regulation
constraints. When the condition $\muint[j]{(k)}
= \wspan[j]{(k)} \cdot \Q^j$ holds for some interval $B^j$, we say
that $B^j$ is \emph{saturated}. The set of all the unsaturated
intervals is computed at Line~\ref{alg:unsat}, and their respective
memory rates $r^j$ given the current assignment $\muint[j]{(k)}$ is
computed at Line~\ref{alg:rates}.

For each iteration, at Line~\ref{alg:argmax} the algorithm selects the
interval $B^j$ with the highest slope for the currently assigned
memory rate $r^j$ among all the unsaturated intervals -- in case two
intervals have the same slope, the tie can be broken
arbitrarily. Finally, at Line~\ref{alg:assignment} the value of memory
transactions $\muint[p]{(k)}$ assigned to the selected interval $B^p$
is modified to the minimum of two expressions: (1) $\mu - \sum_{j \ne
p} \muint[j]{(k)}$, that is, all remaining transactions. In this case,
after the assignment, it holds that $\sum_{j=1}^N \muint[j]{(k)}
= \mu$ and the algorithm terminates immediately. (2)
$\nextfn{p}(r^p) \cdot \wspan[p]{(k)}$, that is,
$\muint[p]{(k)}$ is incremented so that $r^p = \muint[p]{(k)} /
\wspan[p]{(k)}$ becomes equal to the next segment start point.

\begin{framed}
\begin{lstlisting}[caption={Stall maximization over multiple intervals}, 
label=alg:concave, abovecaptionskip=-\medskipamount, belowcaptionskip=0.2cm, 
mathescape=true, numbers=left, 
numberstyle=\scriptsize,basicstyle=\footnotesize,escapeinside={(*@}{@*)}]
Input: $B^1, \ldots, B^N$ (*@\hfill@*)/* sequence of intervals */
Input: $\wspan[1]{(k)}, \ldots, \wspan[N]{(k)}$ (*@\hfill@*)/* span in each interval */ 
Input: $\mu$ (*@\hfill@*)/* total memory request */

Output: $\muint[1]{(k)}, \ldots, \muint[N]{(k)}$ (*@\hfill@*)/* memory requests in each interval */

$\forall j: \muint[j]{(k)}$ $\gets$ $0$ (*@\hfill@*)

do: (*@\label{alg:do}@*) 
  (*@\hfill@*) /* consider only unsaturated intervals */
  $\mathcal{B}$ $\gets$ $\{j ~|~ \muint{j}(k) < \wspan[j]{(k)} \cdot \Q^j\}$ (*@\label{alg:unsat}@*) 
  (*@\hfill@*) /* compute current memory rate $r$ on each unsaturated interval */
  $\forall j \in \mathcal{B}: r^j$ $\gets$ $\muint[j]{(k)} / \wspan[j]{(k)}$ (*@\label{alg:rates}@*) 
  (*@\hfill@*) /* find curve $p$ where $r$ yields maximum stall slope */
  $p$ $\gets$ $\textrm{argmax}_{j \in \mathcal{B}} \big\{\slopefn{j}(r^j)\big\}$  (*@\label{alg:argmax}@*) 
  (*@\hfill@*) /* assign as many as possible transactions to this interval */
  $\muint[p]{(k)}$ $\gets$ $\min\big( \mu - \sum_{j \ne p} \muint[j]{(k)}, \textrm{next}^p(r^p) \cdot \wspan[p]{(k)} \big)$ (*@\label{alg:assignment}@*) 
  (*@\hfill@*) /* stop if all $\muint{}$ assigned, or all intervals are saturated */
until ($\sum_{j=1}^N \muint[j]{(k)} = \mu$ or $\forall j: \muint[j]{(k)} = \wspan[j]{(k)} \cdot \Q^j$) (*@\label{alg:until}@*) (*@\hfill@*)
\end{lstlisting}
\end{framed}

Note that the segment start points, $\mu$ and $\wspan[p]{(k)}$ are all
natural numbers; hence, assuming that the values of variables
$\muint[j]{(k)}$ were integer before the assignment at
Line~\ref{alg:argmax}, the new value assigned to $\muint[p]{(k)}$ is
also integer. Furthermore, the new assignment cannot violate the
constraints $\sum_{j =1}^N \muint[j]{(k)} \leq \mu$ or
$\muint[j]{(k)} \leq \wspan[j]{(k)} \cdot \Q^j$, since we use the
minimum of the two expressions. Hence, this shows that the assignment
to variables $\muint[j]{(k)}$ operated by Algorithm~\ref{alg:concave}
is feasible according to the constraints at
Lines~\ref{alg:natural}-\ref{alg:split_memory} of
Algorithm~\ref{alg:opt}. Furthermore, note that
Algorithm~\ref{alg:concave} is guaranteed to terminate after the
assignment at Line~\ref{alg:assignment} selects the first expression,
or after all intervals have been saturated. The number of segment
start points for each function $\Iconc(r)^j$ is $\mathcal{O}(m)$;
hence, the number of iteration of the algorithm is
$\mathcal{O}(N \cdot m)$.

Finally, we show that once the algorithm terminates, the assignment to
variables $\muint[j]{(k)}$ maximizes the cumulative stall
$\mstall{(k)}
= \sum_{j=1}^N \mstall[j]{(k)} = \sum_{j=1}^N \Iconc(\muint[j]{(k)}
/ \wspan[j]{(k)})^j \cdot \wspan[j]{(k)}$, that is, the objective
function in Algorithm~\ref{alg:opt}. This follows from the way
intervals are selected at Line~\ref{alg:argmax}. By contradiction,
assume that there exists a different feasible assignment, call it
$\{\bar{\mu}^1_{(k)}, \ldots, \bar{\mu}^N_{(k)}\}$, such that
$\sum_{j=1}^N \Iconc(\bar{\mu}^j_{(k)}
/ \wspan[j]{(k)})^j \cdot \wspan[j]{(k)}
> \sum_{j=1}^N \Iconc(\muint[j]{(k)}
/ \wspan[j]{(k)})^j \cdot \wspan[j]{(k)}$; then can obtain
$\{\bar{\mu}^1_{(k)}, \ldots, \bar{\mu}^N_{(k)}\}$ by iteratively
modifying $\{\muint[1]{(k)}, \ldots, \muint[N]{(k)}\}$, subtracting
some number of memory transactions, say $\Delta$, from one variable
$\muint[j]{(k)}$ and adding them to another variable
$\muint[p]{(k)}$. Now define:
\begin{align}
\textrm{slope}^p & = \frac{ \Iconc\big( (  \muint[p]{(k)}+\Delta)/\wspan[p]{(k)}  \big)^p  - \Iconc(   \muint[p]{(k)} / \wspan[p]{(k)}  )^p}{\Delta / \wspan[p]{(k)}}, \nonumber \\
\textrm{slope}^j & = \frac{  \Iconc(   \muint[j]{(k)} / \wspan[j]{(k)}  )^j - \Iconc\big( (  \muint[j]{(k)}-\Delta)/\wspan[j]{(k)}  \big)^j}{\Delta / \wspan[j]{(k)}}, 
\end{align}
as the resulting slopes for functions $\Iconc(r)^p_i$ and
$\Iconc(r)^j_i$. Note that the modification to the variables will
increase the cumulative stall by $\textrm{slope}^p \cdot \Delta$ and
reduce it by $\textrm{slope}^j \cdot \Delta$. But because
Line~\ref{alg:argmax} always selects the function with the highest
slope, it must be $\textrm{slope}^p \leq \textrm{slope}^j$; hence, the
cumulative stall cannot increase, a contradiction. In summary, we have
shown the following lemma:
\begin{lemma}
Algorithm~\ref{alg:concave} terminates in a finite number of
steps. Furthermore, the resulting assignment to variables
$\{\muint[1]{(k)}, \ldots, \muint[N]{(k)}\}$ determines a value of
$\sum_{j=1}^N \Iconc(\muint[j]{(k)}
/ \wspan[j]{(k)})^j \cdot \wspan[j]{(k)}$ equal to the value of the
objective function computed by Algorithm~\ref{alg:opt}.
\end{lemma}

\textbf{Computational Complexity}: 
note that each iteration of the algorithm can be easily optimized to
execute in $\mathcal{O}(1)$. The sum of variables $\muint[j]{(k)}$ at
Lines~\ref{alg:assignment},~\ref{alg:until} can be executed in
constant time by keeping the sum in a variable and updating it each
time $\muint[p]{(k)}$ is modified at Line~\ref{alg:assignment}. The
selection at Line~\ref{alg:argmax} can be performed in constant time
by creating a table of segments ordered by slope. Since ordering the
segments then dominates the complexity of the algorithm, this results
in a $\mathcal{O}\big(N \cdot m \cdot \log(N \cdot m) \big)$ time for
Algorithm~\ref{alg:concave}. Note that, the analysis assumes a 
know memory schedule, but, is generic with respect to CPU scheduling. Thus, it 
is applicable to event-triggered CPU schedulers.

\section{Scenario: Budget Assignment and Schedulability Ratios for IMA Partitions}

Integrated Modular Avionics (IMA) systems use time-triggered scheduling of partitions, also known as ARINC 653 scheduling, where each partition is assigned, at compile time, a fixed start time and span in a major cycle i.e., a hyperperiod (H). These partition-level scheduling decisions are stored at compile time resulting in a static CPU schedule, which is repeated every major cycle. 

Our analysis (Section \ref{sec:wcet_n_bud}) works with known memory assignment across cores and known workload parameters. IMA systems are a natural fit, representing a real-world scenario.
We consider a set of IMA partitions with a fixed major cycle and assignment of partitions to cores.
For simplicity, we assume the order of execution of the partitions is known, and we assume that each partition executes once in the major cycle, and that the major cycle is synchronized among cores.  
Our goal is to use our analysis from Section \ref{sec:wcet_n_bud} and perform an empirical evaluation comparing the ratio of schedulable tasksets to generated tasksets, under dynamic memory budget assignment policy against the static budget assignment policies, under a fixed partition execution order on each core. 

In the next Subsections, we describe the setup used to compare the budget assignment policies and the two sets of experiments, one, that varies the number of cores and two, that varies the number of memory intensive partitions in a system.

\subsection{Setup}

\par{\bfseries{IMA Partition Set Generation:}}
For each experiment run, we consider $m$ cores and a set of $4 \times m$ IMA 
partitions, with a fixed major cycle, i.e., hyperperiod (H) of $128 ms$. The 
earliest start time of each partition is set to $t=0$ and the deadline to the 
hyperperiod i.e. $128 ms$. From the perspective of the analysis (Section 
\ref{sec:wcet_n_bud}), each partition is a workload.

We characterize the varying memory demand between partitions as exhibited by avionic applications \cite{Agrawal:ECRTS2017}, using a parameter --- \textit{memory intensity (MI)} --- , that represents the ratio of pure memory demand to the sum of pure processing demand and pure memory demand of a partition under single-core case i.e., no contentions. We then use a bi-modal distribution for $MI$, where each partition either has a HIGH MI mode or a LOW MI mode. The use of two modes is first, consistent with the memory intensity behavior exhibited by partitions in a real avionic application \cite{Agrawal:ECRTS2017}, and second, some partitions perform I/O activity that is memory- intensive.  All HIGH MI mode partitions are randomly assigned an $MI$ value  in the range of $ [0.5, 0.99]$, whereas for LOW MI mode partitions, the $MI$ value range is $[0.001, 0.1]$. We use a parameter \textit{memory intensity ratio} $MIr$ to vary the number of partitions in the HIGH MI mode to that in the LOW MI mode in the system.

Each partition is then randomly assigned a core, such that each core ends up with $4$ partitions. The setup then generates per partition single-core utilization using UUniFast algorithm~\cite{Bini:Springer2005} such that $U$ is the cumulative single-core utilization of each core. 
The parameter $U$ allows varying the cumulative single-core utilization of partitions assigned to a core. Next, the setup generates $E$ and $\mu$ values for each partition based on its single-core utilization and memory intensity (MI) value, assuming no stall. The $E$ and $\mu$ values of each partition respectively represent an aggregated $E$ demand and an aggregated $\mu$ demand of all tasks assigned to it, in line with existing works like~\cite{Nowotsch:ECRTS2014} and~\cite{Agrawal:ECRTS2017}.

\par{\bfseries{System-wide Parameters:}} We use realistic system-wide 
parameters: $L_{max} =  2.4 \times 10^{-6} s $, $P = 1 ms$, resulting in $Q = 
41666$ as described in~\cite{SCE_ecrts15}.

\par{\bfseries{Budget Assignment Policies:}} 
We consider two static budget assignment policies: \textit{Static and even 
(SE)} that assigns to each core a constant and identical budget of $1 / m $ 
times the total budget, e.g., for a 4-core system $\mathcal{Q} = $\{10416, 
10416, 10416, 10416\}, and  \textit{static and uneven (SU)} that assigns to 
each core a constant budget based on the weight of each core, e.g., 
$\mathcal{Q} = \{416, 20416, 5416, 15416\}$. For the SU policy, we use a 
heuristic to generate the weight of each core and thereby, a budget assignment, 
based on the input partition set.

The key idea behind the heuristic is to assign cores with higher memory demand 
a higher memory budget. The heuristic computes weight of each core based 
on the ratio of the remaining 
cumulative $\mu$ to that of the sum of remaining cumulative $\mu$ and 
cumulative $E$ on a core. Then, the memory bandwidth 
$Q$ is partitioned among cores based on the computed weights resulting in a 
budget assignment. This is 
similar to the term memory intensity, albeit on a core-level.

\begin{figure}[!htb]
		\centering
		\includegraphics[width=\linewidth]{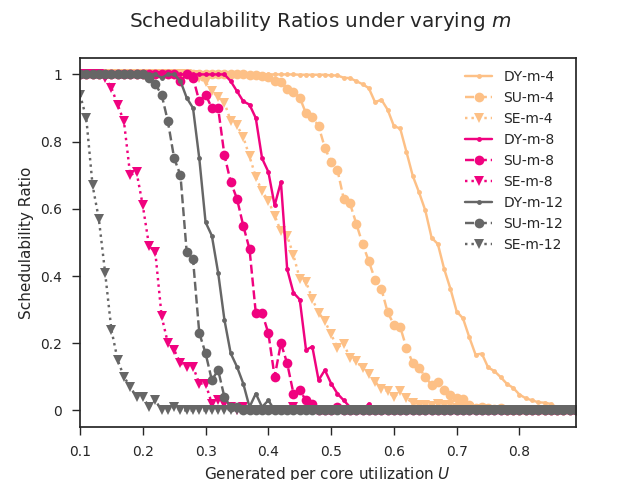}
		\caption{Schedulability ratios under varying $m$ from $4$ to $12$ cores 
		in steps of $4$, for memory budget assignment policies: Dynamic (DY), 
		Static Uneven (SU), and Static Even (SE). Memory intensity ratio $MIr$ 
		of $0.25$} 								\label{fig:exp-vary-m}
\end{figure}

We compare the static policies (SE and SU) against a \textit{dynamic policy 
(DY)}, which 
assigns dynamic memory budgets to cores, using the heuristic. As compared to 
the static SU policy that uses the heuristic at time $t=0$ only, DY policy 
recomputes the weight of every core each time a partition finishes execution, 
resulting in a dynamic budget assignment.

\subsection{Varying Number of Cores $m$}
Figure \ref{fig:exp-vary-m} compares the schedulability ratios for each of the 
three budget assignment policies --- DY, SU and SE ---- under varying the 
number of cores $m$ from $4$ to $12$ in steps of $4$, for a fixed memory 
intensity ratio $MIr$ of $0.25$. In Figure \ref{fig:exp-vary-m}, for each value 
of $U$, we generated $1000$ partition sets for $m=4$ case, and $100$ partition 
sets for each of $m = 8$ and $m=12$ cases. On 
the x-axis, we vary the cumulative per core utilization $U$ from $0.1$ to $0.9$ 
in steps of $0.01$.
 
First, we observe that as the number of cores $m$ increases, the schedulability ratio decreases for the plots shift towards the left, for each of the three budget assignment policies. This is because, with increasing the number of cores, the total memory supply remains constant, albeit the total memory demand increases as the number of HIGH MI mode partitions increase in the system.
Second, for each value of $m$, the dynamic policy DY dominates the static policies SU and SE.

\subsection{Varying Memory Intensity ratio $MIr$}
Now, we vary the memory intensity ratio $MIr$ from $0.15$ to $0.50$ that impacts the number of HIGH MI partitions in the system, and consequently, the number of LOW MI partitions in the system. We set the number of cores $m$ to $8$.

Figure \ref{fig:exp-varying-MIR} shows the schedulability ratios for each of the three budget assignment policies --- DY, SU and SE ---- under varying $MIr$.
On the x-axis, we vary the cumulative per core utilization $U$ from $0.1$ to $0.9$ in steps of $0.01$.
In Figure \ref{fig:exp-varying-MIR}, we generated $100$ partition sets for every combination of $U$ and $MIr$.

As the $MIr$ ratio increases, the cumulative memory load from all cores on the memory increases, in general. Consequently, we observe that schedulability ratio plots shift towards the left on increasing the $MIr$ ratios. 
Further, for each memory intensity ratio $MIr$, the dynamic budget assignment policy DY dominates static policies SU and SE.

\begin{figure}[!htb]
	\centering
	\includegraphics[width=\linewidth ]{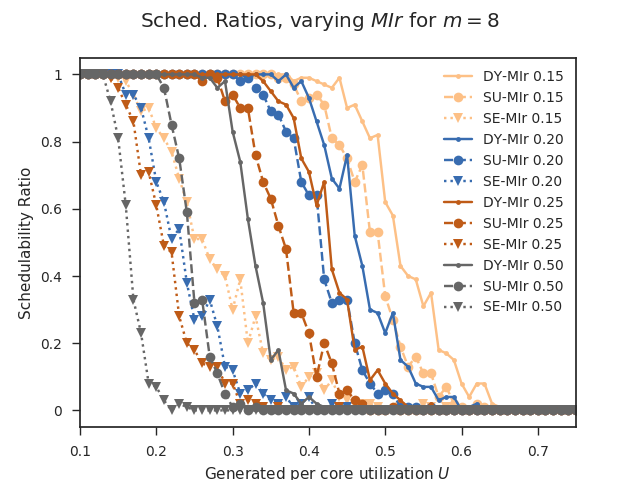}
	\caption{Schedulability ratios under varying memory intensity ratio $MIr$ from $0.15$ to $0.50$, for memory budget assignment policies: Dynamic (DY), Static Uneven (SU), and Static Even (SE). The number of cores $m$ is set to $8$.}  			\label{fig:exp-varying-MIR}
\end{figure}

\section{Related Work}
\label{sec:rel_wrk}

Recent literature on the design of real-time systems on multi-core
platforms considers main memory as a significant source of unpredictability,
and an important interfering channel to mitigate. Predictable memory
controllers have been proposed in~\cite{pret_mc, predator,
  medusa}. OS-level techniques implementable on COTS hardware to
regulate access of cores to main memory have been proposed and
evaluated in~\cite{memguard, memguard_tc, mc_2, mc2_1-out-of-m}. Yet
another body of work has investigated the idea of strictly serializing
access of cores to main memory. For instance, the work
in~\cite{mem_centric} clusters memory operations in tasks via cache
pre-fetching using compiler-level transformations, defining memory-
and execution- phases. Then, a central scheduler only allows at most
one memory-phase to be active at any point in time. A similar scheme
was adopted in~\cite{spm_os, pellizz_dyn_spm:ecrts13,
  dma_spm_pellizz:ecrts17, aer_model} using DMAs instead of
CPU-initiated pre-fetches and scratchpad memories. A recent
work~\cite{Andersson:2018:SAT} proposes an analysis for co-running
tasks contending for memory resources, i.e. with no explicit bandwidth
partitioning.

By clustering and serializing access to shared memory, interference is
avoided by design. Compared to this approach, regulation has the
advantage of being entirely implementable at OS-level. For memory
regulation techniques, analytic bounds for the temporal behavior of
tasks was also derived~\cite{gang_mg_analysis, SCE_ecrts15,
  SCE_ecrts17, memserv_pellizz}. Similarly, the work
in~\cite{behnam:2013} derives runtime guarantees when both a CPU
server and memory regulation are used. These works focus on static
memory bandwidth partitioning.

With respect to static and even budget assignment, a first analysis
was derived in~\cite{SCE_ecrts15}. In~\cite{gang_mg_analysis}, an
analysis for static and uneven bandwidth partitioning was performed
assuming only knowledge of the memory budget $q_i$ for the core under
analysis, and assuming arbitrary assignment to the other $m-1$ cores.
More recently, the work in~\cite{SCE_ecrts17} demonstrated that by
leveraging \emph{exact} knowledge of each core's budget $q_i$ it is
possible to drastically reduce the pessimism of the analysis.

A few works~\cite{memguard, memguard_tc} proposed unused budget
reclamation. However, no offline guarantees can be provided on the
dynamic portion of the assigned budget. The work in~\cite{wang:SIES14}
considers budget reclamation and derives WCET guarantees assuming full
knowledge of the workload on all cores. In a more recent work,
Nowotsch et al.~\cite{Nowotsch:RTNS2013, Nowotsch:ECRTS2014} consider
avionics temporal partitions with pre-defined budget assignment. In
this way, they are able to compute offline the WCET of application
inside a partition, albeit the budget may vary at the boundaries of
partitions. Finally, the work in~\cite{Agrawal:ECRTS2017} relaxes the
strict single budget-to-partition assignment in~\cite{Nowotsch:ECRTS2014} and
allows different budgets being assigned to a partition offline, enabling 
dynamic budget assignment, from a set of design-time fixed budgets. By assuming 
that memory stall is pre-computed in each budget, the WCET  computation problem 
is then be decomposed in (1) assigning ``compatible''  budgets across cores; 
and (2) minimizing the use of high-budget slots by the task under analysis.

{\bf What sets this work apart is the generality of the provided
  results}. Unlike the aforementioned literature, we do not assume any
specific budget re-assignment scheme. In fact, we provide a
methodology that can be used to compute the worst-case runtime of a
task given any dynamic budget-to-core assignment. To use our results,
either exact knowledge of budget assignment over time is known; or a
critical instant for memory budget re-assignments should be
identified.
 
\section{Conclusion}
\label{sec:conc}

In this paper, we presented a methodology to analyze the worst-case
execution time and schedulability of real-time workload under dynamic
memory scheduling. We first introduced a simple iterative algorithm to
compute the span of workload under static and uneven budget-to-core
assignment. We then generalized the problem to consider a generic
memory schedule and formulated the worst-case span analysis as a
stall-maximization problem. Next, we demonstrated that the problem has
strong similarities with concave optimization and proposed a
low-complexity solution to determine the access pattern that maximizes
the overall memory stall. As a use case, we considered an IMA setting
where a subset of partitions run memory-intensive workload. In this
scenario, dynamic memory scheduling outperformed traditional static
bandwidth partitioning. The analysis assumes a known memory schedule. It 
is, however, generic with respect to CPU scheduling. Thus, it is applicable for 
event-triggered CPU schedulers. As a future work, we intend to study online
bandwidth scheduling strategies for which a critical instant on
decisions taken of both processor and memory can be identified.

\section*{Acknowledgment}
The authors would like to thank the anonymous reviewers for their helpful suggestions.

\bibliographystyle{IEEEtran}
\bibliography{ms.bbl}
\appendix

\section{Appendix}

\begin{lemma} \label{lm:curve_above_zero}
Consider a function $f: \mathbb{R}_{\geq 0} \rightarrow \mathbb{R}$ such that $f$ is concave and $f(0) = 0$. Then $\forall a, b \in \mathbb{R}_{\geq 0}$ with $b > a$:
\begin{equation}
f(a) - a \cdot \frac{f(b) - f(a)}{b - a} \geq 0.
\end{equation}
\end{lemma}
\begin{IEEEproof}
By definition of concave function, $\forall x, y \in \mathbb{R}_{\geq 0}$ and $\alpha \in [0,1]$ it holds:
\begin{equation} \label{eq:concave}
f\big( (1-\alpha) \cdot x + \alpha \cdot y \big) \geq (1 - \alpha) \cdot f(x) + \alpha \cdot f(y).
\end{equation}
Substituting $x=0, y = b$ and $\alpha = \frac{a}{b}$ in Equation~\ref{eq:concave} and given $f(0) = 0$ we have:
\begin{equation} \label{eq:fa}
f(a) \geq \Big(1 - \frac{a}{b}\Big) \cdot f(0) + \frac{a}{b} \cdot f(b) = \frac{a}{b} \cdot f(b).
\end{equation}

Finally, using Equation~\ref{eq:fa} we obtain:
\begin{align}
& f(a) - a \cdot \frac{f(b) - f(a)}{b - a} = \nonumber \\
& \frac{f(a)\cdot b - f(a)\cdot a - f(b)\cdot a +f(a)\cdot a}{b-a} = \nonumber \\
& \frac{f(a)\cdot b  - f(b)\cdot a }{b-a} \geq \nonumber \\
& \frac{\frac{a}{b} \cdot f(b)\cdot b  - f(b)\cdot a }{b-a} = 0, 
\end{align}
which yields the hypothesis.
\end{IEEEproof}

\begin{lemma} \label{lm:monotonicI}
Consider a function $g: \mathbb{R}_{\geq 0} \rightarrow \mathbb{R}$ such that $g$ is concave and $g(0) = 0$. Then $\forall \mu \in \mathbb{R}_{\geq 0}, \forall x \in \mathbb{R}_{> 0}$:  $x \cdot g(\mu / x)$ is monotonic non decreasing in $x$.
\end{lemma}
\begin{IEEEproof}
We have to show that $\forall x_1, x_2 \in \mathbb{R}_{>0}$ with $x_2 > x_1$:
\begin{equation}
x_2 \cdot g(\mu / x_2) \geq x_1 \cdot g(\mu / x_1).
\end{equation}

Let us define $K = \frac{ g(\mu/x_1) - g(\mu/x_2) } {1/x_1 - 1/x_2}$. We first show that: 
\begin{equation} \label{eq:g_K}
g(\mu/x_2) - K/x_2 \geq 0.
\end{equation}
Define $f(y) = g(\mu \cdot y), a = 1/x_2, b = 1/x_1$. Note that it holds $b > a > 0, f(0) = g(\mu \cdot 0) = 0$, and since $g$ is concave, $f$ is also concave. Then we obtain by substitution:
$g(\mu/x_2) - K/x_2 = f(a) - a \cdot \frac{f(b) - f(a)}{b - a}$,
which by Lemma~\ref{lm:curve_above_zero} is greater than or equal to 0. 

Finally, using Equation~\ref{eq:g_K} we obtain:
\begin{align}
&x_2 \cdot g(\mu/x_2) = x_2 \cdot \big( g(\mu/x_2) - K/x_2 \big) + K \geq \nonumber \\
&x_1 \cdot \big( g(\mu/x_2) - K/x_2 \big) + K = \nonumber \\
&x_1 \cdot \big( g(\mu/x_2) - K/x_2 + K/x_1 \big) = \nonumber \\
&x_1 \cdot \big( g(\mu/x_2) + K \cdot (1/x_1 - 1/x_2) \big) = \nonumber \\
&x1 \cdot \big( g(\mu/x_2) + g(\mu/x_1) - g(\mu/x_2) \big) = x_1 \cdot g(\mu/x_1),
\end{align}
completing the proof.
\end{IEEEproof}

\begin{IEEEproof}[Proof of Theorem \ref{theorem1}]
  We show that the iteration in Equation~\ref{eq:iteration_static} is a special 
  case of Equation~\ref{eq:iteration_dynamic}, which implies that 
  Theorem~\ref{thm:wcet_single} follows as a corollary of 
  Theorem~\ref{thm:length_dynamic}. $\wspan{(0)}$ is computed in the same way, 
  so we reason about the expression for $\wspan{(k)}$.
  
  Note that the static budget scenario in Section~\ref{sec:wcet_1_bud} is 
  equivalent to a dynamic scenario where there is only one memory scheduling 
  interval $B^1$ of unbounded length. Hence, we have $\wspan[1]{(k-1)} = 
  \wspan{(k-1)}$, and we can define $\Iconc(r) = \Iconc(r)^1$ and $\Q = \Q^1$ 
  without loss of generality. The stall term in 
  Equation~\ref{eq:iteration_dynamic} is thus equal to:
  \begin{align} 
  & \sum_{j=1}^N \Iconc(\muint[j]{(k-1)} / \wspan[j]{(k-1)})^j \cdot 
  \wspan[j]{(k-1)} = \nonumber \\
  & \Iconc(\muint[1]{(k-1)} / \wspan{(k-1)})  \cdot \wspan{(k-1)}. 
  \label{eq:equivalence_sd}
  \end{align}
  Next, consider Algorithm~\ref{alg:opt}; since there is only one variable 
  $\muint[1]{(k-1)}$, the constraints at 
  Lines~\ref{alg:WQ},~\ref{alg:split_memory} are equivalent to: 
  $\muint[1]{(k-1)} \leq \wspan{(k-1)} \cdot \Q$ and $\muint[1]{(k-1)} \leq 
  \mu$. But since increasing the value of the variable cannot cause the 
  objective function to decrease, it follows that the assignment 
  $\muint[1]{(k-1)} = \min(\mu, \wspan{(k-1)} \cdot \Q)$ must maximize the 
  stall. Substituting this value into Equation~\ref{eq:equivalence_sd} yields: 
  \begin{align}
  & \Iconc\big( \min(\mu, \wspan{(k-1)} \cdot \Q)/ \wspan{(k-1)} \big)  \cdot 
  \wspan{(k-1)} = \nonumber \\
  & \Iconc\big( \min(\mu / \wspan{(k-1)}, \Q) \big)  \cdot \wspan{(k-1)}, 
  \end{align}
  which is the stall term in Equation~\ref{eq:iteration_static}. This shows 
  that the two iterations over $\wspan{(k)}$ are identical, completing the 
  proof.
  
\end{IEEEproof} 
\end{document}